\documentclass[preprint,authoryear,12pt]{elsarticle}

\usepackage{amssymb}
\usepackage{amsmath}

\usepackage{psfrag}

\DeclareMathOperator{\cost}{cost} \DeclareMathOperator{\dist}{dist}

\newcommand{\abs}[1]{\lvert #1 \rvert}
\newcommand{\norm}[1]{\lVert #1 \rVert}

\newtheorem{theorem}{Theorem}
\newtheorem{lemma}[theorem]{Lemma}
\newproof{proof}{Proof}

\journal{Computational Geometry: Theory and Applications}

\begin{document}

\begin{frontmatter}



\title{Integer Point Sets Minimizing Average Pairwise
  $L_1$~Distance:
 What is the Optimal Shape of a Town?}


\author[MIT]{Erik D. Demaine}
\author[TUBS]{S\'andor P. Fekete}
\author[FUB]{G\"{u}nter Rote}
\author[TUBS]{Nils Schweer}
\author[FUB]{\hspace{.45cm}Daria Schymura}
\author[GUF]{Mariano Zelke}

\address[MIT]{Computer Science and Artificial Intelligence Lab, MIT, USA.}
\address[TUBS]{Algorithms Group, Braunschweig University of Technology, Germany.}
\address[FUB]{Institut f\"ur Informatik, Freie Universit\" at Berlin, Germany.}
\address[GUF]{Institut f\"ur Informatik, Goethe-Universit\"at, Frankfurt am Main, Germany.}

\begin{abstract}
An {\em $n$-town}, $n\in \mathbb{N}$, is a group of $n$ buildings,
each occupying a distinct position on a 2-dimensional integer grid.
If we measure the distance between two buildings along the
axis-parallel street grid, then an $n$-town has optimal shape if the
sum of all pairwise Manhattan distances is minimized. This problem
has been studied for {\em cities}, i.e., the limiting case of very
large $n$. For cities, it is known that the optimal shape can be
described by a differential equation, for which no closed-form
solution
is known. We show that optimal $n$-towns
can be computed in $O(n^{7.5})$ time. This is also practically
useful,
as it allows us to compute optimal solutions up to $n=80$. 
\end{abstract}

\begin{keyword}
Manhattan distance \sep average pairwise distance \sep integer
points \sep dynamic programming
\end{keyword}

\end{frontmatter}



\section{Introduction}\label{sec:introduction}
Selecting an optimal set of locations is a fundamental problem, not
just in real estate, but also in many areas of computer science.
Typically, the task is to choose $n$ sites from a given set
of candidate locations; the objective is to pick a set that minimizes
a cost function, e.g., the average distance
between sites. As described below, there is a large variety of related results,
motivated by different scenarios.

In general, problems of this type are hard, even
to approximate, as the problem of finding a clique of given size is a special case. Some of the
natural settings have a strong geometric flavor, so it is conceivable
that more positive results can be achieved by exploiting additional
structures and properties.
However, even seemingly easy
special cases are still surprisingly difficult. Until now, there was no
complexity result (positive or negative) for the scenario in which
the candidate locations correspond to the full integer grid, with distances measured
by the Manhattan metric (an \emph{$n$-town}).
Indeed, for the shape of area~1 with minimum
average $L_1$ distance (the ``optimal shape of a city'', arising
for the limit case of $n$ approaching infinity), no simple
closed-form solution is known, suggesting that finding sets of $n$ distinct
grid points (the ``optimal shapes of towns'') may not be an easy task.
This makes the problem mathematically challenging; in addition,
the question of choosing $n$ grid positions with minimum
average $L_1$ distance comes up naturally in grid computing, so
the problem is of both practical and theoretical interest.

In this paper, we give the first positive result by describing
an $O(n^{7.5})$ algorithm for computing sets of $n$ distinct grid points
with minimum average $L_1$ distance. Our method is based on dynamic programming,
and (despite of its relatively large exponent) for the first time allows computing optimal
towns up to $n=80$.

\subsection{Related Work}
\paragraph{Grid Computing}
In grid computing, allocating a task requires selecting $n$ processors from
a given grid,
and the average communication overhead corresponds to
the average Manhattan distance between processors;
\cite{mache96,mache97} and 
\cite{leung02}
propose various metrics for
measuring the quality of a processor allocation, including the
average number of communication hops between processors.
\cite{leung02} considered the problem of allocating processors
 on Cplant, a Sandia National Labs
supercomputer;
they applied and evaluated a scheme
based on space-filling curves, and
 they concluded that the average pairwise Manhattan
distance between processors is an effective metric to optimize.

\paragraph{The Continuous Version}
Motivated by the problem of storing records in a 2-dimensional
array, 
\cite{KarpMW75} studied strategies that minimize
average access time between successive queries; among other results,
they described an optimal solution for the continuous version of our
problem: What shape of area $1$ minimizes the average Manhattan
distance between two interior points? Independently,
\cite{Bender03whatis} solved this problem in the setting of a city,
inspiring the subtitle of this paper. The optimal solution is
described by a differential equation, and no closed-form solution is
known.

\begin{figure}[tb]
\begin{center}
  \includegraphics[width=.73\columnwidth]{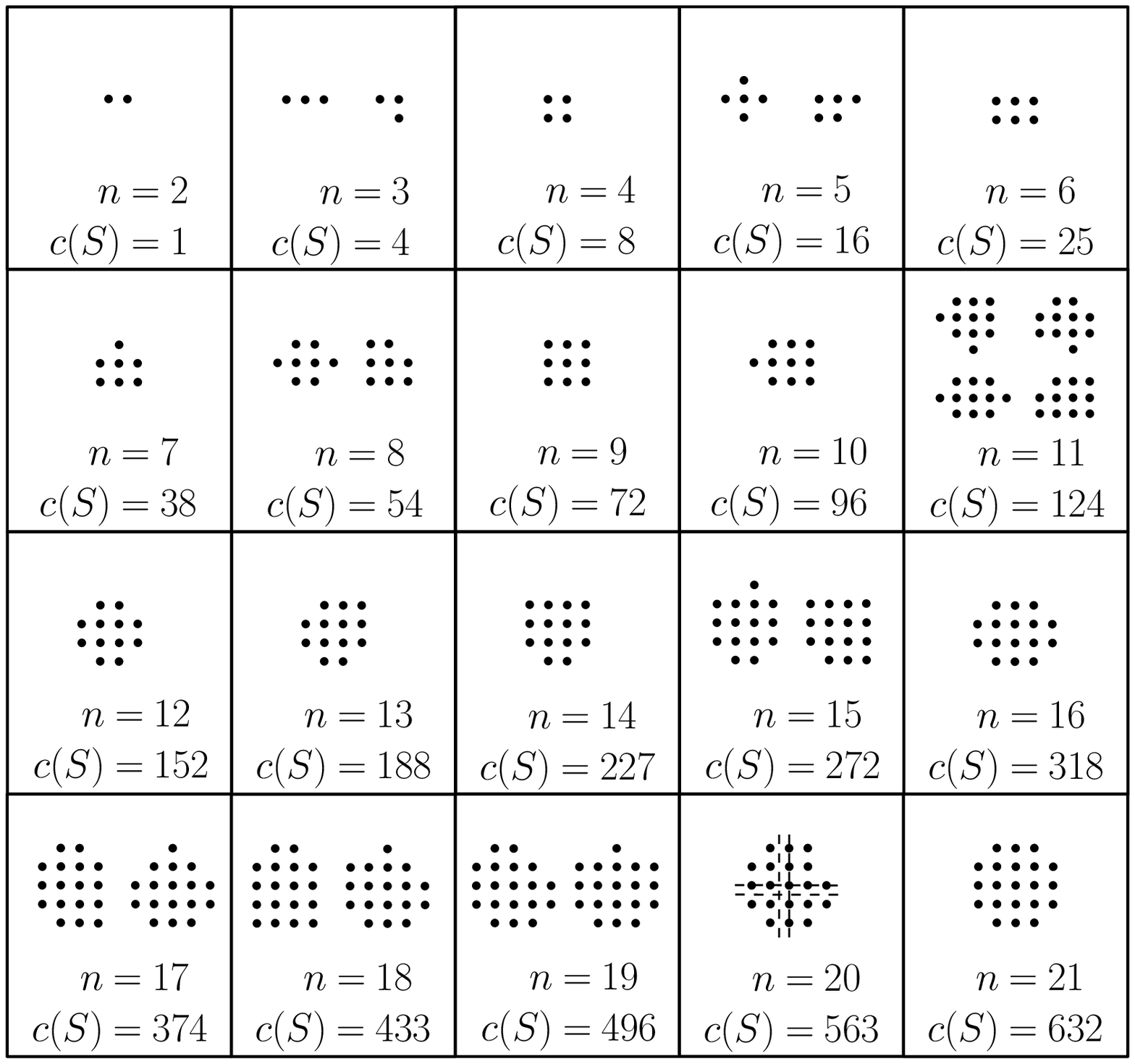}
  \caption{Optimal towns for $n=2,\ldots,21$.  All optimal solutions
    are shown, up to symmetries; the numbers indicate the total
    distance between all pairs of points. The picture for $n=20$ also
    contains the symmetry axes from Lemma~\ref{lem:symmetry}. }
\label{fig:smalltowns}
\end{center}
\end{figure}

\paragraph{Selecting $k$ points out of $n$}
\cite{krumke97} consider the discrete problem of selecting a subset
of $k$ nodes from a network with $n$ nodes to minimize their average
pairwise distance. They prove a $2$-approximation for metric
distances and prove hardness of approximation for arbitrary
distances.
\cite{bender08} solve the geometric version
of this problem, giving an efficient processor allocator for the Cplant setting
described above, and a polynomial-time approximation scheme (PTAS) for minimizing the
average Manhattan distance.
For the reverse problem of {\em maximizing} the
average Manhattan distance, see~\cite{fm-mdgmwc-04}.

\paragraph{The $k$-median Problem}
Given two sets $D$ and $F$, the $k$-median problem asks to choose a
set of $k$ points from $D$ to minimize the average distance to the
points in $F$. For $k=1$ this is the classical {\em Fermat-Weber
problem}.
\cite{Fekete00,fmb-cfwp-05} considered the city-center problem: for
a given city, find a point that minimizes the average Manhattan
distance.
They proved NP-hardness for general $k$ and
gave efficient algorithms for some special cases.

\paragraph{The Quadratic Assignment Problem}
Our problem is a special case of the quadratic assignment problem
(QAP): Given $n$ facilities, $n$ locations, a matrix containing the
amount of flow between any pair of facilities, and a matrix
containing the distances between any pair of locations. The task is
to assign every facility to a location such that the cost function
which is proportional to the flow between the facilities multiplied
by the distances between the locations is minimized. For a survey
see \cite{loiola07}. The cost function in our problem and in the
QAP are the same if we define the distances as the Manhattan
distances between grid points and if we define all flows to be one.
The QAP can not be approximated within any polynomial factor unless
$P=NP$; see~\cite{sahni76}.
\cite{hls-amqap-09}
considered the metric version of this problem with the flow matrix
being a 0/1 incidence matrix of a graph. They state some
inapproximability results as well as a constant-factor approximation
for the case in which the graph has vertex degree two for all but
one vertex.

\paragraph{The Maximum Dispersion Problem}
The reverse version of the discrete problem, where the goal is to
maximize the average distance between points, has also been studied:
In the maximization version, called the \emph{maximum dispersion
problem}, the objective is to pick $k$ points from a set of size $n$
so that the pairwise distance is maximized. When the edge weights
need not obey the triangle inequality, 
\cite{kp-cds-93} give an $O(n^{0.3885})$-approximation. 
\cite{aitt-gfds-00} improve this guarantee to a constant factor in
the special case when $k=\Omega(n)$ and 
\cite{akk-ptasd-99} give a PTAS when $|E|=\Omega(n^2)$ and
$k=\Omega(n)$.

When the edge weights obey the triangle inequality, 
\cite{rrt-hscad-94} give a $4$-approximation that runs in $O(n^2)$
time and 
\cite{hrt-aamd-97} give a $2$-approximation that runs in
$O(n^2+k^2\log k)$ time. For points in the plane and Euclidean
distances, 
\cite{rrt-hscad-94} give an approximation with performance bound
arbitrarily close to $\pi/2\approx 1.57$. For Manhattan distances,
\cite{fm-mdgmwc-04} give an optimal algorithm for fixed $k$ and a
PTAS for general $k$. Moreover, they provide a
$(\sqrt{2}+\varepsilon)$-Approximation for Euclidean distances.

\paragraph{The Min-Sum $k$-Clustering Problem}
Another related problem is called \emph{min-sum $k$-clustering} or
\emph{minimum $k$-clustering sum}. The goal is to separate a graph
into $k$ clusters to minimize the sum of pairwise distances between
nodes in the same cluster.  For general graphs,
\cite{sahni76} show that this problem is NP-hard to approximate
within any constant factor for $k\geq 3$. In a metric space the
problem is easier to approximate:
\cite{guttmannbeck98} give a $2$-approximation,
\cite{indyk99} gives a PTAS for $k=2$, and
\cite{bartal01} give an
$O(1/\epsilon\log^{1+\epsilon}n)$-approximation for general $k$.

\subsection{Our Results} We solve the $n$-town problem with an
$O(n^{7.5})$-time dynamic-programming algorithm.
Our algorithm is based on some properties of an optimal town:
an optimal $n$-town is \emph{convex} in the sense that it contains
all grid points within its convex hull (Lemma~\ref{lem:convex}).
It lies symmetric with respect to a horizontal and a vertical
symmetry axis within a tolerance of $\pm1$ that is due to parity
issues
 (Lemma~\ref{lem:symmetry}).
Furthermore, it fits in an $O(\sqrt n)\times O(\sqrt n)$ square
 (Lemma~\ref{bounded-width}).
 We also present
computational results and discuss the relation between
the optimum continuous cities and their discretized counterparts
($n$-towns, and $n$-block cities).

\section{Properties of Optimal Towns}\label{sec:prelim}
We want to find a set of $n$ distinct points from the integer grid
$\mathbb{Z} \times \mathbb{Z}$ such that the sum of all pairwise
Manhattan distances is minimized. A set $S \subset \mathbb{Z} \times
\mathbb{Z}$ of cardinality $n$ is an \emph{$n$-town}. An $n$-town
$S$ is \emph{optimal} if its \emph{cost} 
\begin{equation}
  \label{eq:c}
c(S) := \frac12 \cdot
\sum_{s \in S}
\sum_{t \in S}
 \|s-t\|_1
\end{equation}
is minimum.
Figure~\ref{fig:smalltowns} shows solutions for small $n$ and their
cost, and Table~\ref{table:smalltowns} in Section~\ref{sec:results}
 shows optimal cost values
$c_{\textrm{town}}(n)$ for $n\le80$.
 We define the \emph{$x$-cost} $c_x(S)$ as
$\sum_{\{s,t\} \in S \times S}|s_x-t_x|$, where $s_x$ is the
$x$-coordinate of $s$; \emph{$y$-cost} $c_y(S)$ is the sum of all
$y$-distances, and $c(S) = c_x(S) + c_y(S)$. For two
sets $S$ and $S'$, we define $c(S,S')=\sum_{ \{s,s'\}\in S\times
S'}\|s-s'\|_1$. If $S$ consists of a single point $t$, we write
$c(t,S')$ instead of $c(\{t\},S')$ for convenience. A town $S$ is
\emph{convex} if the set of grid points in the convex hull of $S$
equals $S$.







In proving various properties of optimal towns, we will often make a
local modification by moving a point $t$ of a town to a different
place~$r$. The next lemma expresses the resulting cost change.

\begin{lemma}\label{lem:changecostfunc}
Let $S$ be a town, $t\in S$ and $r\notin S$. Then,
$$
c( (S\setminus t)\cup r) = c(S) - c(t,S) + c(r,S) - \|r-t\|_1 \; .
$$
\end{lemma}
\begin{proof}
Let $p$ be a point in $S$. Then, its distance to $t$ is $\|t-p\|_1$
and the distance to $r$ is $\|r-p\|_1$. Hence, the change in the
cost function is $\|r-p\|_1 - \|t-p\|_1$. We need to subtract
$\|r-t\|_1$ from the sum over all points in $S$ because $t$ is
removed from~$S$.
\qed
\end{proof}

\begin{lemma}\label{lem:convex}
An optimal $n$-town is convex.
\end{lemma}

The following proof holds in any dimension and with any norm for measuring
the distance between points.
\begin{proof}
We prove that a nonconvex $n$-town $S$ cannot be optimal.
Take a grid point $x \notin S$ in the convex hull of $S$.
Then there are points $x_1,x_2,\ldots,x_k \in S$ such that
$x =\lambda_1 x_1 + \lambda_2 x_2 +\cdots + \lambda_k x_k$ for some
$\lambda_1, \lambda_2,\ldots, \lambda_k \ge 0$ with $\sum \lambda_i
= 1$. Because every norm is a convex function, and the sum of convex
functions is again convex, the function $f_S(x)=c(x,S)=\sum_{s \in
S} \|x-s\|_1$ is convex. Therefore, $$f_S(x) \le \lambda_1 f(x_1) +
\lambda_2 f(x_2)+\cdots + \lambda_k f(x_k),$$ which implies $f_S(x)
\le f_S(x_i)$ for some $i$. Using Lemma~\ref{lem:changecostfunc} we
get
$$c((S\setminus x_i) \cup x) = c(S) - f_S(x_i) + f_S(x) - \|x-x_i\|_1
<c(S).$$
 This means that $S$ is not optimal.
\qed
\end{proof}

Obviously, if we translate every point from an $n$-town by the same
vector, the cost of the town does not change. We want to distinguish
towns because of their shape and not because of their position
inside the grid and, therefore, we will only consider optimal towns
that are placed around the origin. Lemma~\ref{lem:symmetry} makes
this more precise:
an optimal $n$-town is roughly
symmetric with respect to a vertical and a horizontal symmetry line,
see Fig.~\ref{fig:dynamicprogram} for an illustration.
 Perfect symmetry is not possible
since some rows or columns may have odd length and others even length.

We need some notation before: For an $n$-town
$S$, the $i$-th \emph{column} of $S$ is the set $C_i= \{\, (i,y) \in S
: y \in \mathbb{Z}\,\}$ and the $i$-th \emph{row} of $S$ is the set
$R_i= \{\, (x,i) \in S : x \in \mathbb{Z}\,\}$.

\begin{lemma}[Symmetry]
\label{lem:symmetry}
In every optimal $n$-town $S$, the centers of all rows of odd length
lie on a common vertical grid line $V_o$. The centers of all rows of even
length lie on a common line $V_e$ that has distance $\frac12$ from
$V_o$. A corresponding statement holds for the centers of odd and
even columns that lie on horizontal lines $H_o$ and $H_e$ of
distance~$\frac12$. Moreover, without changing its cost, we can
place $S$ such that $H_o$ and $V_o$ are mapped onto the $x$-axis and
$y$-axis, respectively, and $H_e$ and $V_e$ lie in the negative
halfplanes.
\end{lemma}


\begin{proof}
For a row $R_j$ and $r\in \mathbb{Z}$, let $R_j+(r,0)$ be
the row $R_j$ horizontally translated by $(r,0)$. If two rows $R_i$ and $R_j$
are of the same parity, a straightforward calculation
(using Lemma~\ref{lem:changecostfunc}) shows that
the cost $c(R_i,R_j+(r,0))$ is minimal if
and only if the centers of $R_i$ and $R_j+(r,0)$ have the same $x$-coordinate. If the parities differ, $c(R_i,R_j+(r,0))$ is
minimized with centers having $x$-coordinates of distance 1/2.
The total cost is
\begin{eqnarray*}
c(S) = c_x(S) + c_y(S)  =
 \sum_{i,j} \sum_{s \in R_i, t \in R_j} |s_x -t_x| + c_y(S)
\end{eqnarray*}
If we translate every row $R_i$ of $S$ horizontally by some
$(r_i,0)$, $c_y(S)$ does not change. The solutions that minimize
$\sum_{s \in R_i, t \in R_j} |s_x + r_i -t_x + r_j|$ for all $i,j$
simultaneously are exactly those that align the centers of all rows
of even length on a vertical line $V_e$ and the centers of all rows
of odd length on a vertical line $V_o$ at offset $\frac12$ from
$V_e$. The existence of the lines $H_o$ and $H_e$ follows
analogously.

We can translate $S$ such that $H_o$ and $V_o$ are mapped onto the
$x$- and the $y$-axis and rotate it by a multiple of $90^\circ$
degrees such that $H_e$ and $V_e$ lie in the negative halfplanes.
These operations do not change $c(S)$. \qed
\end{proof}

{From} the convexity statement in Lemma~\ref{lem:convex} (together
with Lemma~\ref{lem:symmetry}) we know that $C_0$ is the largest
column, and the column lengths decrease to both sides, and similarly
for the rows.
Our algorithm will only be based on this weaker property
(\emph{orthogonal convexity}); it will not make use of convexity
\emph{per se}.
We will, however, use convexity 
one more time to prove that the
lengths of the columns are $O(\sqrt n)$, in order to reduce the
running time.

In the following we assume the symmetry property of the last lemma.
%
%
%
%
%
 For an $n$-town $S$, let the \emph{width} of $S$ be
$w(S)=\max_{i \in \mathbb{Z}} |R_i|$ and the \emph{height} of $S$ be
$h(S)=\max_{i \in \mathbb{Z}} |C_i|$. We will now show that the width
and the height cannot differ by more than a factor of~2.
Together with convexity, this will imply that they are bounded by
$O(\sqrt n)$ (Lemma~\ref{bounded-width}).

\begin{lemma}
\label{relative-bound-width} For every optimal $n$-town $S$,
$$w(S) > h(S)/2 -3 \mbox{ and }h(S) > w(S)/2 - 3.$$
\end{lemma}

\begin{proof}
Let $S$ be an $n$-town, $w=w(S)$, and $h=h(S)$, and assume
 $w\le h/2-3$.
 Let $t=(0,l)$ be the
topmost and $(k,0)$ be the rightmost point of $S$,
with
$l=\lfloor \frac{h-1}2\rfloor$
and
$k=\lfloor \frac{w-1}2\rfloor$. Let $r=(k+1,0)$.
We show that $c((S \setminus t) \cup r) < c(S)$,
and
thus, $S$ is not optimal.
By Lemma~\ref{lem:changecostfunc},
 the change in cost is $c(r,S)-c(t,S)-|k+l+1|$.
We show that $c(r,S)-c(t,S)\le0$ by
 calculating this difference column by column.
This proves then that replacing $t$ with $r$ yields a gain
of at least $|l+k+1| \ge 1$, and we are done.
%
Let us calculate the difference
 $c(r,C_j) - c(t,C_j)$ for a column $C_j$ of height $|C_j|=s\le h$:
\begin{align*}
 c(r,C_j)-c(t,C_j)
&= \sum_{i=-\lceil \frac{s-1}{2} \rceil}^{\lfloor \frac{s-1}{2} \rfloor}\! (|i|-(l-i))
 + s(k+1-j-|j|)
\\&= \sum_{i=0}^{\lfloor \frac{s-1}{2} \rfloor} (i-(l-i)) +
 \sum_{i=1}^{\lceil \frac{s-1}{2} \rceil} (i-(l+i))
 + s(k+1-j-|j|)
\\&= 2 \sum_{i=0}^{\lfloor \frac{s-1}{2} \rfloor}\! i
 -sl
 + s(k+1-j-|j|)
\\&
 \le \tfrac{(s-1)(s+1)}{4} -s\tfrac{h-2}2 + s(k+1)
 \le \tfrac{s^2}{4} -s\tfrac{h-2}2 + s\cdot\tfrac{w+1}2
\\&
 = \tfrac{s}{4}(s-2h+2w+6)
 \le \tfrac{s}{2}(-h+2w+6)
\le 0
\end{align*}
%
\qed
\end{proof}

\begin{lemma}
\label{bounded-width} For every optimal $n$-town we have
$$\max \{ w(S), h(S)  \} \le 2 \sqrt{n} + 5.$$
\end{lemma}

\begin{proof}
Let $w=w(S)$ and $h=h(S)$. Assume without loss of generality
that $h\ge w$.
We know from Lemma
\ref{relative-bound-width} that $w > h/2 - 3$.
By Lemma \ref{lem:symmetry}, we choose a topmost, a rightmost,
a bottommost, and a leftmost point of $S$ such that the convex hull
of these four points is a quadrilateral with a vertical and horizontal diagonal,
approximately diamond-shaped.
Let $H$ be the set of all grid points contained in this quadrilateral.
The area of the quadrilateral equals $(w-1)(h-1)/2$, and its boundary
contains at least 4 grid points.
Pick's theorem says that the area of a simple grid polygon equals
the number of its interior grid points $H_i$ plus half of the number of
the grid points $H_0$ on its boundary minus 1.
This implies $|H| = |H_i|+|H_0|
=
(|H_i|+|H_0|/2-1) + |H_0|/2+1
 \ge  (w-1)(h-1)/2+3$.
Because of Lemma~\ref{lem:convex}, all points in $H$ belong to $S$.
Since $H$ consists of at most $n$ points, we have
$$
n \geq |H| \ge (w-1)(h-1)/2 + 3 >
(h/2 - 4)(h-1)/2 + 3
$$ 
Solving the equation
 $
 h^2-9h+20-4n
 = 0$ shows that
$$h \le 2
  \sqrt{n+ 1/16} + 9/2 \le 2\sqrt n + 5.
\eqno\hbox{\qed}$$
\end{proof}

\section{Computing Optimal Solutions}\label{sec:optimal}
\begin{figure}[tb]
\begin{center}
\includegraphics[height=6cm]{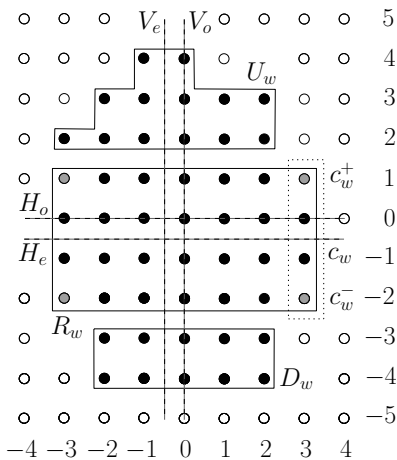} \caption{The
lines $V_o$, $V_e$, $H_o$, and $H_e$ from Lemma~\ref{lem:symmetry}.
The rectangle $R_w$ and the set of points above and below it with
cardinality $U_w$ and $D_w$, respectively. The gray points are the
corner points of $R_w$. In this example, the height $c_w$ of column
$w$ is set to $c=4$.} \label{fig:dynamicprogram}
\end{center}
\end{figure}

We will now describe a dynamic-programming algorithm for computing
optimal towns. A program for this algorithm is listed
in~\ref{code}.

We denote by $c_i=|C_i|$ the number of selected points in column $i$
and by $c_i^+$ and $c_i^-$ the row index of the topmost and
bottommost selected point in $C_i$, respectively. We have $c_i =
c_i^+ - c_i^- +1$; see Fig.~\ref{fig:dynamicprogram}.

\begin{lemma}\label{lem:rectangle}
Let $S$ be an optimal $n$-town \textup(placed as described in
Lemma~\ref{lem:symmetry}\textup) containing the points $(i,c_{i}^+)$ and
$(i,c_{i}^-)$, for $i\geq 0$. Then all points inside the rectangle
$[-i,i]\times[c_{i}^-, c_{i}^+]$ belong to $S$.

Similarly, if $S$ contains
the points $(-i,c_{-i}^+)$ and $(-i,c_{-i}^-)$, for $i\geq1$, then it contains all
points in the rectangle
$[-i,i-1]\times[c_{-i}^-, c_{-i}^+]$.
\end{lemma}
\begin{proof}
If $(i,c_{i}^+)$ and $(i,c_{i}^-)$ are contained in $S$ then, by
Lemma~\ref{lem:symmetry}, $(-i,c_{i}^+)$ and $(-i,c_{i}^-)$ belong
to $S$ as well. By Lemma~\ref{lem:convex} all points inside the
convex hull of these four points are contained in $S$. The same
arguments hold for the second rectangle.  \qed
\end{proof}

Now we describe the dynamic program. It starts with the initial
empty grid and chooses new columns alternating from the set of
columns with nonnegative and with negative column index, i.e.,
in the order
$0,-1,1,-2,2,\ldots$. Let $w \geq 0$ be the index of the currently
chosen column and fix $c_w$ to a value~$c$. We describe the dynamic
program for columns with nonnegative index; columns with negative
index are handled similarly.
(In the program that is described in the appendix, we use a trick to
avoid dealing with negative columns: they are mapped to columns with positive
index by reflecting everything at the $y$-axis, with a proper
adjustment to take into account that the placement of
Lemma~\ref{lem:symmetry} is not invariant under this transformation.)

We know from Lemma~\ref{lem:rectangle} that in every optimal
solution, every point inside the rectangle $R_w =
[-w,w]\times[c_w^-, c_w^+]$ is selected.
We define
$$\cost(w,c,\Delta_{w}^{\mathrm{UR}},
\Delta_{w}^{\mathrm{DR}}, \Delta_{w}^{\mathrm{UL}},
\Delta_{w}^{\mathrm{DL}},U_w,D_w)$$
as the minimum cost of a town
with columns $-w,\ldots,w$ of height $c_i\ge c$ for $-w\le i\le w$
and $c_w=c$ where $U_w$ points lie above the rectangle $R_w$, having
a total distance $\Delta_w^{\mathrm{UL}}$ and
$\Delta_w^{\mathrm{UR}}$ to the upper-left and upper-right corner of
$R_w$, respectively, and $D_w$ points lie below $R_w$, having a
total distance $\Delta_w^{\mathrm{DL}}$ and $\Delta_w^{\mathrm{DR}}$
to the lower-left and lower-right corner of $R_w$. For a given $n$,
we are looking for the $n$-town with minimum cost where $(2w+1)c +
U_w + D_w = n$. Next we show that
$\cost(w,c,\Delta_{w}^{\mathrm{UR}}, \Delta_{w}^{\mathrm{DR}},
\Delta_{w}^{\mathrm{UL}}, \Delta_{w}^{\mathrm{DL}},U_w,D_w)$ can be
computed recursively.

Consider the current column $w$ with $c_w=c$. The cost from all
points in this column to all points above $R_w$, in $R_w$, and below
$R_w$ can be expressed as
\begin{multline*}
  \sum_{k=c^-}^{c^+}(\Delta_{w}^{\mathrm{UR}} + (c^+ - k)\cdot U_w)
+\sum_{i=-w}^{w}\sum_{j=c^-}^{c^+}\sum_{k=c^-}^{c^+}[(w-i)+|k-j|]
 \\
+ \sum_{k=c^-}^{c^+}(\Delta_{w}^{\mathrm{DR}} + (k - c^-)\cdot|D_w|
)  .
\end{multline*}
We can transform this into
\begin{multline}\label{eq:dist}
 c\cdot(\Delta_{w}^{\mathrm{UR}}+\Delta_{w}^{\mathrm{DR}}+U_w\cdot c^+ -D_w\cdot c^-)
+  c^-\cdot(D_w-U_w)\cdot\left((c+1)\bmod 2\right) \\
+
\left(c^2w+\frac{c^3-c}{3}\right)\cdot(2w+1)-\frac{c^3-c}{6}
,
\end{multline}
which, obviously, depends only on the parameters $w$, $c$,
$\Delta_{w}^{\mathrm{UR}}$, $\Delta_{w}^{\mathrm{DR}}$, $U_w$, and
$D_w$ (the two parameters $\Delta_{w}^{\mathrm{UL}}$,
$\Delta_{w}^{\mathrm{DL}}$ are needed if we consider a column with
negative index). We denote the expression \eqref{eq:dist} by $\dist(w,
c,\Delta_{w}^{\mathrm{UR}}, \Delta_{w}^{\mathrm{DR}},\allowbreak
\Delta_{w}^{\mathrm{UL}},\allowbreak \Delta_{w}^{\mathrm{DL}},\allowbreak U_{w},\allowbreak D_{w})$ and
state the recursion for the cost function:
\begin{multline}\label{eq:recursion}
 \cost(w,c,\Delta_{w}^{\mathrm{UR}},\ldots,\Delta_{w}^{\mathrm{DL}},U_{w},D_{w}) \\
 = \min_{c_{-w}\geq c} \{ \cost(-w, c_{-w},
\Delta_{-w}^{\mathrm{UR}},\ldots,
\Delta_{-w}^{\mathrm{DL}},U_{-w},D_{-w}) \}  \\
 + \dist(w,c,\Delta_{w}^{\mathrm{UR}},\ldots,
\Delta_{w}^{\mathrm{DL}},U_{w},D_{w})
\end{multline}
By Lemma~\ref{lem:rectangle} it suffices to consider only previous
solutions with $c_{-w}\geq c$. In the step before, we considered the
rectangle $R_{-w}=[-w,w-1]\times[c_{-w}^+,c_{-w}^-]$. Hence, the
parameters with index $-w$ can be computed from the parameters with
index $w$ as follows:
\begin{align*}
U_{-w} & =  U_{w} - 2w\cdot(c_{-w}^+ - c^+), \\
D_{-w} & =  D_{w} - 2w\cdot(c^- - c_{-w}^-), \\
\Delta_{-w}^{\mathrm{UR}} & =   \Delta_{w}^{\mathrm{UR}} -
\sum_{i=-w}^{w}\sum_{j=c^+ + 1}^{c_{-w}^+} \left[(w-i) + (j-c^+)\right]
 \\ & \qquad
- \left[ U_w - U_{-w} \right]\cdot (c_{-w}^+ - c^+ + 1), \\
\Delta_{-w}^{\mathrm{DR}} & =   \Delta_{w}^{\mathrm{DR}} -
\sum_{i=-w}^{w}\sum_{j=c^- - 1}^{c_{-w}^-} \left[(w-i) + (c^- - j)\right] \\
 & \qquad - \left[ D_w - D_{-w} \right]\cdot (c^- - c_{-w}^- + 1). \\
\end{align*}
The parameters $\Delta_{-w}^{\mathrm{UL}}$ and
$\Delta_{-w}^{\mathrm{DL}}$ can be computed analogously and the cost
function is initialized as follows:
$$
\cost(0,c,0,0,0,0,0,0) =
\begin{cases}
\frac{c^3-c}{6}, & \text{if } 0 \le c \le 2\sqrt{n}+5, \\
\infty, & \mathrm{otherwise.}
\end{cases}
$$
The bound on $c$ has been shown in Lemma \ref{bounded-width}.

\begin{theorem}
\label{running}
An optimal $n$-town can be computed by dynamic programming in
$O(n^{15/2})$ time.
\end{theorem}
\begin{proof}
  We have to fill an eight-dimensional array
  $\cost(w,c,\Delta^\mathrm{UR},\Delta^\mathrm{DR},\Delta^\mathrm{UL},$
  $\Delta^\mathrm{DL},U,D)$.  Let
  $C_{\max}$ denote the maximum number of occupied rows and columns in
  an optimum solution.  By Lemma~\ref{bounded-width}, we know that
  $C_{\max}=O(\sqrt n)$.

  The indices $w$ and $c$ range over an interval of size
  $C_{\max}=O(\sqrt n)$. Let us consider a solution for some fixed $w$
  and $c$.  The parameters $U$ and $D$ range between 0 and~$n$.  However,
  we can restrict the difference between $U$ and $D$ that we have to
  consider: If we reflect the rectangle
  $R=[-w,w]\times[c^-,c^+]$  about its horizontal
  symmetry axis, the $U$ points above $R$ and the $D$ points below $R$
  will not match exactly, but in each column, they differ by at most
  one point, according to Lemma~\ref{lem:symmetry}. It follows that $\abs
  {U-D}\le C_{\max}=O(\sqrt n)$.  (If the difference is larger, such a
  solution can never lead to an optimal $n$-town, and hence we need
  not explore those choices.)  In total, we have to consider
  only $O(n\cdot\sqrt n)=O(n^{3/2})$ pairs $(U,D)$.

  The same argument helps to reduce the number of quadruples
  $(\Delta^\mathrm{UL},\Delta^\mathrm{UR},$ $\Delta^\mathrm{DL},\Delta^\mathrm{DR})$.
  Each $\Delta$-variable
  can range between 0 and $n\cdot 2C_{\max}= O(n^{3/2})$.
  However, when reflecting around the horizontal symmetry axis of $R$, each of the
at most $D_{\max}$ differing points contributes at most $2C_{\max}=
O(\sqrt n)$ to the difference between the distance sums
  $\Delta^\mathrm{UL}$ and $\Delta^\mathrm{DL}$. Thus we have
  $\abs{\Delta^\mathrm{UL}-\Delta^\mathrm{DL}}\le
C_{\max} \cdot 2C_{\max}=O(n)$, and similarly,
  $\abs{\Delta^\mathrm{UR}-\Delta^\mathrm{DR}}=O(n)$.

By a similar argument, reflecting about
 the vertical symmetry axis of $R$, we conclude that
  $\abs{\Delta^\mathrm{UL}-\Delta^\mathrm{UR}}=O(n)$ and
$\abs{\Delta^\mathrm{DL}-\Delta^\mathrm{DR}}=O(n)$. In summary, the
total number of quadruples
  $(\Delta^\mathrm{UL},\Delta^\mathrm{UR},\Delta^\mathrm{DL},\Delta^\mathrm{DR})$
that the algorithm has to consider is $O(n^{3/2})\cdot O(n)\cdot
O(n)\cdot O(n)=O(n^{9/2})$. In total, the algorithm processes
$O(\sqrt n)\cdot O(\sqrt n)\cdot O(n^{3/2})\cdot O(n^{9/2})=O(n^7)$
8-tuples. For each 8-tuple, the recursion~\eqref{eq:recursion} has
to consider at most $C_{\max}= O(\sqrt n)$ values $c_{-w}$, for a
total running time of $O(n^{15/2})$.  \qed
\end{proof}

\section{\boldmath $n$-Towns, Cities, and $n$-Block Cities}\label{sec:townscitiesblockcities}
 For large values of $n$,
 $n$-towns converge
towards the continuous weight distributions of {\em cities}.
However, the arrangement of buildings in many cities are discretized
in a different sense: An {\em{$n$-block city}} is the
union of $n$ axis-aligned unit squares (``city blocks''),
see Fig.~\ref{disctown} below. In the following,
we discuss $n$-block cities, and
we discuss the relation between $n$-towns and $n$-block
cities.



For a planar region $R$, let $c'(R)$ be the integral of 
Manhattan distances between all point pairs in $R$:
$$c'(R) :=
\int_{p\in R}
\int_{q\in R} \norm{p-q}\, dp\, dq
$$
  When $R$ has
area 1, this is the expected distance between two random points
in~$R$.
Scaling a shape $R$ by a factor of $d$ increases the total cost by a
factor of $d^5$, i.e., by a factor of $A^{2.5}$ for an area of $A$.
This motivated \cite{Bender03whatis} to use the expression
$D(R):=\frac{c'(R)}{A(R)^{2.5}}$ 
as a scale-independent measure for the
quality of the shape of a city. For example, a square $Q$ of any side length $a$ gets the same
value
\begin{equation*}
D(Q)=
\frac{1}{a^{5}} \left(
a^2\cdot\!
 \int_0^a\!\! \int_0^a\! |x_1-x_2|\, dx_2 dx_1+
a^2\cdot\!
 \int_0^a\!\! \int_0^a\! |y_1-y_2|\, dy_2 dy_1\right) =
\frac{2}{3} \; .
\end{equation*}
A circle $C$ yields $D(C)=\frac{512}{45\pi^{2.5}}\approx 0.6504$ and
the optimal shape achieves a value of $\psi=0.650\,245\,952\,951\ldots$

We will consider $n$-block cities $Q(S)$ consisting of unit
squares (``blocks'') centered
at a set of $n$ grid points $S \subset \mathbb{Z} \times
\mathbb{Z}$.
 We denote a unit square centered at point
$s=(s_x,s_y)$ by $Q(s) = [s_x-\frac12,s_x+\frac12]
\times [s_y-\frac12,s_y+\frac12]$, and then we have
$$Q(S) := \bigcup_{s\in S} Q(s).$$
The average distance $D(Q(S))$ of an $n$-block city can be decomposed
into average distances between blocks:
$$c'(Q(S))
=
\int\limits_{p\in Q(S)}
\int\limits_{q\in Q(S)} \norm{p-q}_1\, dp\, dq
=
\sum_{s \in S}
\sum_{t \in S}
\int\limits_{p\in Q(s)}
\int\limits_{q\in Q(t)} \norm{p-q}_1\, dp\, dq
$$
Using the notation
$$
d'(s) :=
\int_{p\in Q(s)}
\int_{q\in Q(0)} \norm{p-q}_1\, dp\, dq
,$$
we can express this as
\begin{equation}
  \label{eq:c'}
c'(Q(S))
=
\sum_{s \in S}
\sum_{t \in S}
d'(s-t)
.
\end{equation}
The average distance $d'(s)$ between two square
blocks at an offset $s$ can be expressed as follows: If the two blocks
don't lie in the same row or column ($s_x\ne0$ and $s_y\ne0$), the
average distance is simply the distance $\norm s _1$ between the
centers, since positive and negative
deviations from the block centers average out.
When two blocks lie in the same column, then the $y$-distances
average out to the $y$-distance between the centers, but the average $x$-distance is not the
$x$-distance between the centers (which would be 0), but
$
\int_{-1/2}^{+1/2}
\int_{-1/2}^{+1/2}
 |x_1-x_2|\, dx_2 \,dx_1 =\frac13$.
Similarly, for two blocks in the same row, we must add $\frac13$ to
the distance $\norm s _1$ between the centers. Finally, for two
identical blocks, we have already seen that the average distance is
$\frac23$. We can express this compactly as
\begin{align*}
d'(s) = d'((s_x,s_y))
&=
\abs {s_x} + \tfrac13[s_x\ne 0]
+
\abs {s_y} + \tfrac13[s_y\ne 0]
\\&
=
\norm s _1 + \tfrac13[s_x\ne 0]
+\tfrac13[s_y\ne 0]
,
\end{align*}
where the notation $[s_x\ne 0]$ is 1 if the predicate
$s_x\ne 0$ holds and 0 otherwise.
With these conventions, the expression
\eqref{eq:c'}
for the cost $c'(Q(S))$ of an $n$-block city $Q(S)$
looks very similar
to
\eqref{eq:c}
for the cost $c(S)$ of a town $S$,
except for
the correction terms $\frac 13$ in the summands and for the factor
$\frac12$.
The factor $\frac12$ accounts for the fact that in the sum $c(S)$ of a
town, each {pair} of (distinct) points is counted once, whereas in the
integral $c'(R)$, each pair of points is considered twice, as two
ordered pairs.
To make these expressions better comparable, we introduce the factor
$\frac12$ and define 
\begin{equation}
\nonumber 
 c_{\mathrm{city}}
(S) :=
\frac 12 \cdot c'(Q(S))
=
\frac12\cdot
\sum_{s \in S}
\sum_{t \in S}
d'(s-t)
.
\end{equation}

The new ``distance'' $d'$ is not a norm (for example, $d'(0)\ne0$),
but 
the properties from Section~\ref{sec:prelim} remain true for this
new objective function.  Thus, the dynamic programming formulation of
Section~\ref{sec:optimal} can be adapted to compute optimal $n$-block
cities.
As we shall now demonstrate,
all lemmas of Section~\ref{sec:prelim} hold verbatim
for $n$-block cities, except that the formula for
the change incurred by moving a single block to a different place
(Lemma~\ref{lem:changecostfunc}) must be adapted:
\begingroup
\renewcommand\thetheorem{\ref{lem:changecostfunc}$\mathit{'}$}
\begin{lemma}
  Let $S\subset \mathbb{Z}^2$
be the set of centers of an $n$-block city, $t\in S$ and $r\notin S$. Then,
$$
c_{\mathrm{city}}( (S\setminus t)\cup r) = c_{\mathrm{city}}(S) - c_{\mathrm{city}}(t,S) + c_{\mathrm{city}}(r,S) - d'(r-t),
$$
where
$c_{\mathrm{city}}(p,S):= \sum_{q\in S} d'(p-q)$.
\end{lemma}
\endgroup
The proof is the same as for
Lemma~\ref{lem:changecostfunc}. One can easily check
that the distance $d'$ from $t$ to itself
and the distance from $r$ to itself
are correctly accounted for.
\qed

Convexity of the optimum solution (Lemma~\ref{lem:convex}) holds
true for $n$-block cities. The proof goes through almost verbatim.
The expression $d'(p)$ is not a norm, but it is a convex function of
$p$, and this is all that is needed.
Approximate symmetry of the optimal solution (Lemma~\ref{lem:symmetry})
remains true. The calculations (which have not been shown in detail
anyway), are modified, but the conclusion is the same.

When comparing an $n$-town $S$ and a corresponding $n$-block city
$Q(S)$, we have to add 1/6 for each pair of blocks that are in the
same column, and for each pair of blocks that are in the same row.
 Using the notation $r_j$ and $c_i$ of Section~\ref{sec:optimal}
for row and column lengths,
and writing
$c_{\textrm{town}}(S)$ instead of $c(S)$ for improved clarity,
 we get thus:
\begin{equation}
  \label{eq:lambda}
\textstyle
c_{\textrm{city}}(S)-c_{\textrm{town}}(S) =
\Lambda(S):=\frac{1}{6} (\sum_i c_i^2 +\sum_j r_j^2)
\end{equation}
This {\em adjustment term} accounts for the discretization effect.
For example, a $1$-town has an average distance of 0, as all the weight is
concentrated in a single point, while a $1$-block city has
 an average distance of $2/3$,
just like any other square,
 (and a $c_{\textrm{city}}$ value of $1/3$).

Using this expression, it is easy to show that the bound of 2 on the
 aspect ratio holds for $n$-block cities in the same form as in
 Lemma~\ref{relative-bound-width}.
The proof of Lemma~\ref{relative-bound-width} establishes that
$c_{\textrm{town}}(S)$ decreases when a top-most point $t$ from the
longest column of height $h$ is added to the right of the
 longest row of length $w$.
For an $n$-block city, the
adjustment term
$\Lambda(S)$ decreases by at least $h^2-(h-1)^2$ when $t$ is removed,
and it increases by $(w+1)^2-w^2 + 1^2$ when $r$ is added.
Under the assumption of the proof ($w\le h/2-3$), the total change
of $\Lambda$
is negative, and the modified solution is an improvement also when
$S$ is regarded as an $n$-block city.

From Lemma~\ref{relative-bound-width} we conclude that the bound of
$2\sqrt n+5$ on the height and width of optimal $n$-block cities
(Lemma~\ref{bounded-width}) holds as well.  (The argument of the proof
of Lemma~\ref{bounded-width} is purely geometric: it is based on
convexity and does not use the objective function.)

%
Thus, we conclude that the adjustment term~\eqref{eq:lambda} is
asymptotically bounded by $\Lambda(S) = \Theta(n^{1.5})$.

When considering the continuous weight distributions of $n$-block
cities, we have to account twice for each pair of discrete block
centers; hence, the appropriate measure for the quality of an
$n$-town is $\Phi(S)={2c(S)}/{n^{2.5}}$. The measure for the
corresponding $n$-block city is
$\Psi(S)={2(c_{\mathrm{town}}(S)+\Lambda(S))}/{n^{2.5}}$. Thus, the relative
difference is $\Theta(\frac{1}{n})$.

\begin{figure}[t]
\begin{center}
\input{discretetown.pstex_t} \caption
{Values of $\Phi$ and $\Psi$ for 
optimal $n$-block cities
 for $n=2,\ldots,21$.  All optimal solutions
    are shown, up to symmetries.
These are simultaneously shapes of optimal $n$-towns,
 but 
 for $n=3,11,15,17,18,19$, there are additional
 $n$-towns that are tied for the optimum
(with the same value $\Psi$), cf.~Figure~\ref{fig:smalltowns}.
}
 \label{disctown}
\end{center}
\end{figure}

\begin{figure}[t]
\begin{center}
  \includegraphics[width=.73\columnwidth]{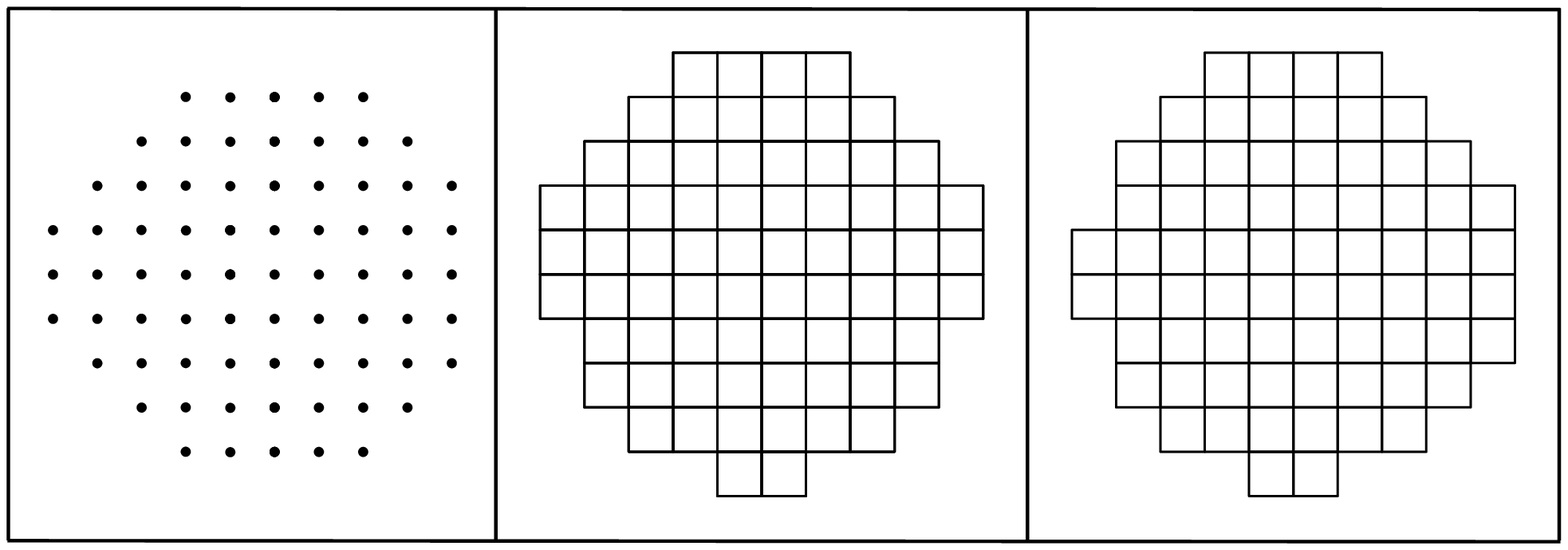}
  \caption{The optimal $n$-town and the two optimal $n$-block cities for $n=72$.}
\label{fig:72}
\end{center}
\end{figure}

In Figure~\ref{disctown}, we show the corresponding values for the
small examples from Figure~\ref{fig:smalltowns}.
Figure~\ref{optsols} shows results for some larger values $n$.
One can observe how $\Phi(S)$ and $\Psi(S)$ converge from below and above
towards the optimal city value $\psi$ of about 0.650\,245\,952\,951. 
Note that convergence is not monotone, neither for $\Phi$ nor for $\Psi$.

\section{Computational Results}
\label{sec:results}

Optimal $n$-towns and optimal $n$-block cities do not necessarily
have the same shape. For $n\le 21$, a comparison of
Figures~\ref{fig:smalltowns} and~\ref{disctown} shows that, while
not all optimal $n$-towns are optimal $n$-block cities, every
optimal $n$-block city is simultaneously an optimal $n$-town.
However, this is not always true. 
In fact, for $n=72$, there are two shapes for
 optimal $n$-block cities, none of which is optimal for an $n$-town,
see Figure~\ref{fig:72}.
This is the only instance of this phenomenon up to $n=80$, but we
surmise it will be more and more frequent for larger~$n$.
We have more comments on this phenomenon at the end of this section.

We have calculated the optimal costs
$c_{\textrm{town}}$ and
$c_{\textrm{city}}$ up to $n=80$ points. The results are shown in
Table~\ref{table:smalltowns}. When there are several optimal solutions
(except symmetries), this is indicated by a star,
together with the multiplicity.

\begingroup
\let\oldarrarystretch = \arraystretch
\renewcommand{\arraystretch}{1.065}
\begin{table}
\centering
{\footnotesize
\begin{tabular}{|r|rr|rr|r||r|rr|rr|r|}
\hline
$n$& $c_{\textrm{town}}$& \ $E_1$ & $c_{\textrm{city}}$& \ \ \ $3E_2$ & $E_3$
& $n$& $c_{\textrm{town}}$& \ \ $E_1$ & $c_{\textrm{city}}$& \ \ \ \ \ $3E_2$ & $E_3$
\\\hline
         1 &     0  &$   0$&      0\rlap{$\,^1\!/_3$}&$   0$ 
        &$   1$
& 41 &  3446  &$   1$&   3530\rlap{$\,^1\!/_3$}&$   2$ 
&$   1$
\\
 2 &     1  &$   0$&      2\rlap{$$}&$   0$ 
&$   1$
& 42 &  3662  &$   1$&   3749\rlap{$$}&$   3$ 
&$   0$
\\
 3 &     4\rlap{$^{*(2)}$}&$   0$&      5\rlap{$\,^2\!/_3$}&$   1$ 
&$   1$
& 43 &  3886  &$   2$&   3976\rlap{$$}&$   5$ 
&$   0$
\\
 4 &     8  &$   0$&     10\rlap{$\,^2\!/_3$}&$  -1$ 
&$   1$
& 44 &  4112  &$  -3$&   4205\rlap{$\,^1\!/_3$}&$ -10$ 
&$   0$
\\
 5 &    16\rlap{$^{*(2)}$}&$   1$&     19\rlap{$\,^2\!/_3{}^{*(2)}$}&$   1$ 
&$   1$
& 45 &  4360\rlap{$^{*(2)}$}&$   6$&   4456\rlap{$\,^2\!/_3$}&$  17$ 
&$   1$
\\
 6 &    25  &$   0$&     30\rlap{$$}&$  -1$ 
&$   1$
& 46 &  4612\rlap{$^{*(2)}$}&$  10$&   4712\rlap{${}^{*(2)}$}&$  31$ 
&$   1$
\\
 7 &    38  &$   0$&     44\rlap{$$}&$   0$ 
&$   1$
& 47 &  4868\rlap{$^{*(2)}$}&$  11$&   4970\rlap{$\,^1\!/_3{}^{*(2)}$}&$  29$ 
&$  -2$
\\
 8 &    54\rlap{$^{*(2)}$}&$   0$&     61\rlap{$\,^1\!/_3{}^{*(2)}$}&$   0$ 
&$   1$
& 48 &  5128  &$   7$&   5234\rlap{$$}&$  18$ 
&$  -1$
\\
 9 &    72  &$  -1$&     81\rlap{$$}&$  -3$ 
&$   2$
& 49 &  5398  &$   4$&   5507\rlap{$\,^1\!/_3$}&$  11$ 
&$  -1$
\\
10 &    96  &$   0$&    106\rlap{$\,^1\!/_3$}&$   0$ 
&$   1$
& 50 &  5675  &$   1$&   5788\rlap{$$}&$   0$ 
&$   0$
\\
\hline
11 &   124\rlap{$^{*(4)}$}&$   2$&    135\rlap{$\,^2\!/_3{}^{*(2)}$}&$   3$ 
&$   0$
& 51 &  5960  &$  -4$&   6076\rlap{$\,^1\!/_3$}&$ -13$ 
&$   0$
\\
12 &   152  &$  -1$&    165\rlap{$\,^1\!/_3$}&$  -4$ 
&$   1$
& 52 &  6248  &$ -14$&   6368\rlap{$$}&$ -43$ 
&$   1$
\\
13 &   188  &$   0$&    203\rlap{$$}&$  -1$ 
&$   1$
& 53 &  6568  &$  -1$&   6691\rlap{$\,^2\!/_3$}&$  -4$ 
&$   1$
\\
14 &   227  &$   0$&    244\rlap{$$}&$  -1$ 
&$   1$
& 54 &  6890  &$   5$&   7017\rlap{$\,^1\!/_3$}&$  15$ 
&$   2$
\\
15 &   272\rlap{$^{*(2)}$}&$   1$&    290\rlap{$\,^2\!/_3$}&$   3$ 
&$   1$
& 55 &  7222\rlap{$^{*(2)}$}&$  12$&   7352\rlap{$\,^1\!/_3$}&$  35$ 
&$   0$
\\
16 &   318  &$  -1$&    338\rlap{$\,^2\!/_3$}&$  -4$ 
&$   1$
& 56 &  7556\rlap{$^{*(2)}$}&$  13$&   7690\rlap{$$}&$  36$ 
&$   0$
\\
17 &   374\rlap{$^{*(2)}$}&$   1$&    396\rlap{$\,^1\!/_3$}&$   3$ 
&$   0$
& 57 &  7896  &$  10$&   8033\rlap{$\,^2\!/_3$}&$  28$ 
&$   0$
\\
18 &   433\rlap{$^{*(2)}$}&$   2$&    457\rlap{$\,^1\!/_3$}&$   5$ 
&$   0$
& 58 &  8243  &$   5$&   8384\rlap{$\,^2\!/_3$}&$  14$ 
&$   1$
\\
19 &   496\rlap{$^{*(2)}$}&$   2$&    522\rlap{$\,^1\!/_3$}&$   4$ 
&$   0$
& 59 &  8604  &$   4$&   8749\rlap{$\,^1\!/_3$}&$  13$ 
&$   1$
\\
20 &   563  &$   0$&    591\rlap{$\,^2\!/_3$}&$   0$ 
&$   1$
& 60 &  8968  &$  -2$&   9117\rlap{$\,^1\!/_3$}&$  -6$ 
&$   2$
\\
\hline
21 &   632  &$  -5$&    663\rlap{$$}&$ -15$ 
&$   1$
& 61 &  9354  &$   3$&   9506\rlap{$\,^1\!/_3$}&$   9$ 
&$   0$
\\
22 &   716  &$   0$&    749\rlap{$\,^1\!/_3$}&$  -1$ 
&$   1$
& 62 &  9749\rlap{$^{*(2)}$}&$   9$&   9904\rlap{$\,^2\!/_3{}^{*(2)}$}&$  24$ 
&$  -1$
\\
23 &   804\rlap{$^{*(2)}$}&$   2$&    839\rlap{$\,^1\!/_3$}&$   6$ 
&$   1$
& 63 & 10146  &$   7$&  10305\rlap{$\,^1\!/_3$}&$  17$ 
&$  -2$
\\
24 &   895  &$   2$&    933\rlap{$$}&$   7$ 
&$   2$
& 64 & 10556  &$   8$&  10719\rlap{$\,^1\!/_3$}&$  21$ 
&$  -1$
\\
25 &   992  &$   2$&   1032\rlap{$\,^1\!/_3$}&$   6$ 
&$   1$
& 65 & 10972  &$   5$&  11139\rlap{$\,^2\!/_3$}&$  15$ 
&$   0$
\\
26 &  1091  &$  -2$&   1134\rlap{$$}&$  -5$ 
&$   2$
& 66 & 11400  &$   5$&  11571\rlap{$\,^2\!/_3$}&$  14$ 
&$   1$
\\
27 &  1204  &$   2$&   1249\rlap{$$}&$   4$ 
&$   1$
& 67 & 11836\rlap{$^{*(2)}$}&$   3$&  12011\rlap{$\,^2\!/_3$}&$   7$ 
&$   1$
\\
28 &  1318  &$   0$&   1365\rlap{$\,^1\!/_3$}&$  -1$ 
&$   0$
& 68 & 12280  &$  -2$&  12460\rlap{$$}&$  -4$ 
&$   2$
\\
29 &  1442  &$   2$&   1492\rlap{$$}&$   5$ 
&$   1$
& 69 & 12728  &$ -12$&  12912\rlap{$\,^1\!/_3$}&$ -34$ 
&$   3$
\\
30 &  1570  &$   1$&   1622\rlap{$\,^2\!/_3$}&$   4$ 
&$   1$
& 70 & 13209  &$   1$&  13396\rlap{$\,^2\!/_3$}&$   2$ 
&$   1$
\\
\hline
31 &  1704  &$   0$&   1759\rlap{$\,^2\!/_3$}&$   1$ 
&$   2$
& 71 & 13700\rlap{$^{*(3)}$}&$  13$&  13891\rlap{$$}&$  37$ 
&$  -1$
\\
32 &  1840  &$  -6$&   1898\rlap{$\,^2\!/_3$}&$ -16$ 
&$   3$
& 72 & 14193  &$  17$&  14388\rlap{${}^{*(2)}$}&$  49$ 
&$  -1$
\\
33 &  1996  &$   1$&   2057\rlap{$$}&$   4$ 
&$   2$
& 73 & 14690  &$  15$&  14888\rlap{$$}&$  40$ 
&$  -4$
\\
34 &  2153  &$   3$&   2216\rlap{$\,^2\!/_3$}&$   8$ 
&$   1$
& 74 & 15195  &$  11$&  15397\rlap{$\,^1\!/_3$}&$  27$ 
&$  -4$
\\
35 &  2318  &$   5$&   2384\rlap{$$}&$  12$ 
&$   0$
& 75 & 15712  &$   8$&  15918\rlap{$\,^1\!/_3$}&$  17$ 
&$  -4$
\\
36 &  2486  &$   3$&   2554\rlap{$\,^2\!/_3$}&$   6$ 
&$  -1$
& 76 & 16232  &$  -3$&  16442\rlap{$\,^2\!/_3$}&$ -14$ 
&$  -4$
\\
37 &  2656  &$  -5$&   2727\rlap{$\,^2\!/_3$}&$ -16$ 
&$  -1$
& 77 & 16780  &$   4$&  16995\rlap{$$}&$   7$ 
&$  -3$
\\
38 &  2847  &$   1$&   2921\rlap{$\,^2\!/_3$}&$   3$ 
&$   0$
& 78 & 17335  &$   7$&  17554\rlap{$\,^2\!/_3$}&$  18$ 
&$  -2$
\\
39 &  3040  &$   2$&   3117\rlap{$\,^2\!/_3$}&$   5$ 
&$   0$
& 79 & 17904\rlap{$^{*(2)}$}&$  13$&  18128\rlap{$$}&$  37$ 
&$  -2$
\\
40 &  3241  &$   3$&   3322\rlap{$$}&$   9$ 
&$   1$
& 80 & 18478  &$  14$&  18706\rlap{$\,^2\!/_3$}&$  40$ 
&$   0$
\\
\hline
\end{tabular}
}
\caption{Optimal towns $c_{\textrm{town}}$ and $n$-block cities $c_{\textrm{city}}$.
 $^*$ indicates multiple 
 solutions.
}
\label{table:smalltowns}
\end{table}
\let\arrarystretch = \oldarraystretch
\endgroup

It is clear that an optimal $n$-block city is never better than an
optimal continuous city of area $n$ that has a value
$\psi n^{5/2}/2$.
Empirically we found the approximation
$c_{\textrm{city}} \approx
 \psi n^{5/2}/2 + 0.115\cdot n^{3/2}$
with $\psi= 0.650245952951$.
 The order of magnitude of
the ``discretization penalty''
$0.115\cdot n^{3/2} =\Theta(n^{3/2})$
is explained as follows:
changing the continuous city of area $n$ into blocks affects
$\Theta(\sqrt n)$ squares
along the boundary. For each adjustment in one of these squares,
distances to $n$ other squares are affected.

Since $c_{\textrm{city}}$ is a multiple of 1/3, we rounded our estimate to the
nearest multiple of $1/3$ and used the approximation formula
$$\bar c_{\textrm{city}} :=
\lfloor 3\psi n^{5/2}/2 + 0.345\cdot n^{3/2}\rceil/3$$
The notation $\lfloor \cdot\rceil$ denotes rounding to the nearest integer.
Table~\ref{table:smalltowns} shows the error $E_2 := c_{\textrm{city}}-\bar
c_{\textrm{city}}$ of this approximation.  (Actually, the table shows $3E_2$,
which is an integer.)

For $n$-towns, on the other hand, we found the approximation
$c_{\textrm{town}} \approx
 \psi n^{5/2}/2 - 0.205\cdot n^{3/2}$.
So this seems to approximate the
optimal continuous city from below, but we do not have a proof of this fact.

The deviation $E_1$
between $c_{\textrm{town}}$
and its approximation formula
$$\bar c_{\textrm{town}} := \lfloor \psi n^{5/2}/2 - 0.205\cdot n^{3/2}\rceil$$
is shown in
Table~\ref{table:smalltowns}.

Finally, we look at
the difference between $c_{\textrm{city}}$ and $c_{\textrm{town}}$.
For a given point set $S$, it is the quantity $\Lambda$
defined in~\eqref{eq:lambda}.  It is estimated as $0.32\cdot n^{3/2}$.
The table shows the error
 $E_3 := 3\cdot(c_{\textrm{city}}-c_{\textrm{town}})-\lfloor 0.96\cdot n^{3/2}\rceil$.
Apart from the rounding, $E_3$ would equal $3(E_2-E_1)$.

One can see that the error $E_3$ is much smaller than one might
expect from the random-looking fluctuations of $E_1$ and $E_2$.
This can be explained by the fact that the expression~\eqref{eq:lambda}
 for $\Lambda(S)$
is apparently not so sensitive to small deviations of the shape $S$.


\begin{figure}[t]
\begin{center}
\includegraphics[width=0.8\columnwidth]{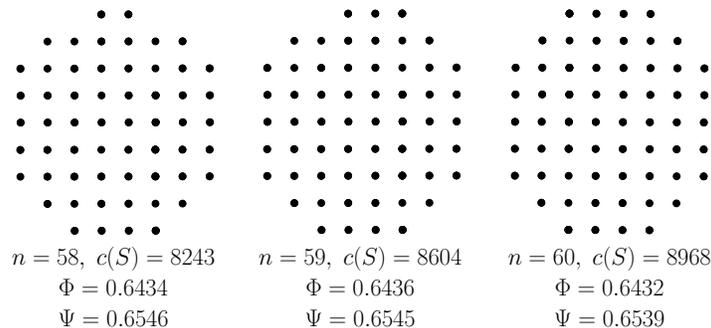} \caption{Optimal
$n$-towns for $n=58,59,60$; these are simultaneously the shapes of
the optimal $n$-block cities.}
 \label{optsols}
\end{center}
\end{figure}

Accordingly, Table~\ref{table:smalltowns} exhibits the tendency
that the deviations of $c_{\textrm{town}}$ and $c_{\textrm{city}}$
``above'' and ``below average'' occur for the same values of~$n$:
$n$-towns and $n$-block cities with the same number $n$ 
 suffer
equally from the effects of discretization.
A glance at the optimal solutions (Figures~\ref{fig:smalltowns}
and~\ref{disctown})
shows that
the costs are
below average when the shapes are highly symmetric, for example
$n=9$, 12, 21, but also $n=60$  (Figure~\ref{optsols}).
On the other hand, when there is no unique ``very good'' shape, one
can expect a greater variation of different solutions that try to come
close to the optimal continuous shape.  Indeed, larger values of $E_1$
and $E_2$ in Table~\ref{table:smalltowns} tend to go hand in hand with
a greater multiplicity of optimal solutions. The worst values of $E_1$
and $E_2$ occur for $n=72$; this is the first value of $n$ where optimal
$n$-block cities and optimal $n$-towns differ (Figure~\ref{fig:72}).
This is probably no coincidence:
when there is a greater variety of solutions that can compete for being
best, the distinction of the objective function between $n$-block
cities and $n$-towns is more likely to make a difference.

\section{Outlook}
\label{sec:outlook}

We have shown that optimal $n$-towns can be computed
in time $O(n^{7.5})$. This is of both theoretical and practical
interest, as it yields a method polynomial in $n$ that also allows extending
the limits of the best known solutions; however,
there are still some ways how the result could be improved.

Strictly speaking, the method is only pseudo-polynomial, as
the input size is $O(\log n)$. It is not clear how
the corresponding output could be described in polylogarithmic
space; any compact encoding would
lead to
 a good and compact
approximation of the optimal (continuous) city curve, for which there is
only a description by a differential equation with no known closed-form solution.
For this reason we are sceptical that a polynomial solution
is possible.


We are more optimistic about lowering the number of parameters
in our dynamic program, and thus the exponent, by exploiting
convexity or
stronger symmetry properties. This may also make it possible to compute
optimal solutions for larger $n$.
One possible avenue could arise if partial solutions
would satisfy some monotonicity property; however,
the unique optimal 9-town is not contained in the unique optimal 12-town.
Thus, there is no way for a town to organically grow
and remain optimal at all times.
Generally,
an optimal $n$-town does not necessarily contain an optimal
 $(n-1)$-town.
(The smallest example occurs for $n=35$,
no optimal 35-town contains an optimal 34-town.)

As discussed in the last section, there is still a variety of questions
regarding the convergence of optimal solutions for growing $n$,
approaching the continuous solution in the limit.
As indicated, we have a pretty good idea how this continuous value is
approximated from below and above by $n$-towns and $n$-block cities;
however, we do not have a formal proof of the lower bound property
of $n$-block cities.

It is easy to come up with good and fast approximation methods:
In the continuous case, even a square is within 2.5\% of the optimal
shape; a circle reaches 0.02\%; consequently, simple greedy heuristics
will do very well. Two possible choices are iteratively adding points
to minimize the total cost, or (even faster) as close as possible to
a chosen center.

As mentioned in the introduction, a closely related, but harder
problem arises when $n$ locations are to be chosen from a given set of
$k>n$ points, instead of the full integer grid. This was studied by \cite{bender08},
who gave a PTAS, but were unable to decide the complexity.
It is conceivable that a refined dynamic-programming approach may
yield a polynomial solution; however, details can expected
to be more involved, so we leave this for future work.
The same holds for other metrics.

Finally, one can consider the problem
in higher dimensions.
A crucial property of our dynamic-programming solution
is that the \emph{interface} between the points in
the columns that have already been constructed and
the points to be added in the future can be characterized by
a few parameters.
A similar property does not hold in three dimensions,
and therefore one cannot extend our
 dynamic-programming approach to higher dimensions.
For the same reason, the Euclidean distance version
cannot be solved by our method, since, unlike in the Manhattan
case, the effect of the $U_w$ points on the upper
side of the current rectangle on the distance to points that are
inserted in the future cannot be summarized in the parameters
$\Delta^{\mathrm{UR}}$
and
$\Delta^{\mathrm{UL}}$.
Moreover, in higher dimensions, there is no known solution for the
continuous case; the corresponding calculus-of-variations problem
will be harder to solve than in two dimensions.

\paragraph{Acknowledgements}
We thank the reviewers for their careful reading and their helpful comments.

\bibliographystyle{elsarticle-harv}
\bibliography{lib}







\appendix

\begin{figure}
  \centering
{\small
\footnotesize
\begin{verbatim}
n_target = 40 # run up to this value of n
\end{verbatim}
\vspace{-1,6\baselineskip}
\begin{verbatim}
cost_array = {} # initialize data for "array"
from math import sqrt
width_limit = int(2*sqrt(n_target)+5)
for w in range(0,width_limit+2):
 for cc in range(width_limit,0,-1):
  cost_array[w,cc]={}
MAX = n_target**3 # "infinity", trivial upper bound on cost
opt = (n_target+1)*[MAX] # initialize array for optimal values

cost_array[0,width_limit][0,0,0,0,0,0]=0 # starting "town" with no columns
for w in range(0,width_limit+1):
 for cc in range(width_limit,0,-1):
  for (D_up_right, D_down_right, D_up_left, D_down_left,
       n_up, n_down), cost in cost_array[w,cc].items():
\end{verbatim}
\vspace{-1,7\baselineskip}
\begin{verbatim}
    D_up_left   += n_up # add 1 horizontal unit to all left-distances
    D_down_left += n_down
\end{verbatim}
\vspace{-1,7\baselineskip}
\begin{verbatim}
    for c in range(cc,-1,-1): # decrease size c of new column one by one
      n = n_up+n_down + (w+1)*c # (w = previous value of w)
      if n <= n_target: # total number of occupied points so far
        new_cost = cost + ( (D_up_left + D_down_left) * c +
                            (n_up + n_down) * c*(c-1)/2 +
                            (c+1)*c*(c-1)/6 * (2*w+1) +
                            c*c * w*(w+1)/2 )
        if c==0: # a completed town
          opt[n] = min( new_cost, opt[n] )
        else: # store cost of newly constructed partial town
          ind = (D_up_left, D_down_left, D_up_right, D_down_right,
                 n_up, n_down) # exchange left and right when storing
          cost_array[w+1, c][ind] = min ( new_cost,
          cost_array[w+1, c].get(ind, MAX) )
      # decrease c by 1:
      if (c%2)==1: # remove an element from the top of the leftmost column
          n_up += w
          D_up_left  += n_up + w*(w+1)/2
          D_up_right += n_up + w*(w-1)/2
      else: # remove from the bottom
          n_down += w
          D_down_left  += n_down + w*(w+1)/2
          D_down_right += n_down + w*(w-1)/2
\end{verbatim}
\vspace{-1,6\baselineskip}
\begin{verbatim}
for n in range(1,n_target+1): print n, opt[n]
\end{verbatim}
}

  \caption{\textsc{Python} program for computing optimal $n$-towns}
  \label{fig:program}
\end{figure}

\section{Program for Computing Optimal Towns}
\label{code}
Figure~\ref{fig:program} shows a short program in the programming language \textsc{Python} that
implements our algorithm.
The program calculates and prints the costs of optimal $n$-towns for
all values of $n$
up to the specified limit $n=\verb:n_target:$.
  Instead of an 8-dimensional array, the
costs are stored 
 as a \emph{dictionary} in the variable
{\tt
\verb:cost_array:[\verb:w:,\verb:cc:][\verb:D_up_right:, \verb:D_down_right:, \verb:D_up_left:, \verb:D_down_left:,
       \verb:n_up:, \verb:n_down:]}.
This makes the program a lot simpler, since we don't have to worry about
allocating arrays with explicit limits, and
 incurs little overhead, since internally, \textsc{Python} dictionaries
are implemented as hash tables, providing constant expected access time.

Instead of adding rows alternately on the left and on the right, the
program always adds a new row on the left side, but (implicitly) reflects
the town about the $y$-axis when storing a cost value, achieving the
same effect.

The main loop of the program does not use the recursion in the
form~\eqref{eq:recursion}, which calculates the optimum cost of a
configuration from all partial solution that lead to it when a column
is added.  Instead, it makes a ``forward'' transfer, generating
all successor configurations of a given configuration.  This has the
advantage that certain ``impossible'' parameter sets are automatically
excluded.  For example, in the running time analysis for
Theorem~\ref{running}, we argued that parameter pairs $U,D$ with
$\abs{U-D}>C_{\max}$ need not be considered.
(The parameters $U$ and $D$ correspond to the variables \verb:n_up: and \verb:n_down:.)
 Since the program only
adds columns which are (approximately) balanced about the $x$-axis,
it will never generate solutions with such parameters.

The program can be adapted for computing optimal $n$-block
\emph{cities}. Then the additional cost $\Lambda$
{from}~\eqref{eq:lambda}
 between blocks in the same row or column must be
taken into account.  One simply has to extend the last line in the
computation of \verb:new_cost::
{\small
\begin{verbatim}
             c*c * w*(w+1)/2 )
\end{verbatim}
}\noindent to
{\small
\begin{verbatim}
             c*c * w*(w+1)/2 ) * 6  + c*c + (cc-c)*w*w
\end{verbatim}
\noindent}The resulting cost is scaled by a factor of~6, but the
end result is then always even, so we could divide it by~2
(cf.~Table~\ref{table:smalltowns}: all values $c_{\textrm{city}}(n)$ are multiples
of~$1/3$).

For $n=40$, the program takes a few seconds, but for $n=80$ it takes
hours.
For larger $n$ the space becomes a more severe bottleneck than
the running time; thus it is important to release
 storage when it is no longer needed, for
example by resetting {\tt \verb:cost_array:[\verb:w:,\verb:cc:]=\verb:{}:}
after each outer loop.
There are several possibilities to speed up the program.
The cost of some approximately circular solution can be taken as an
initial upper bound.  With this upper bound, one can then derive a
stronger bound \verb:width_limit: on the maximum height and width by
{ad-hoc} methods.  During the calculation, one can also prune cost
values that are so large that they cannot possibly lead to a better
solution.  The given program computes only the optimum cost. We have
extended it to also remember the optimal solutions. This program has
133 lines and was used to produce the data of
Table~\ref{table:smalltowns}.

\end{document}